\documentclass[conference,letter]{IEEEtran}
\IEEEoverridecommandlockouts

\usepackage{color}

\usepackage{graphicx,tabularx,array,amsmath,amsthm,thmtools}

\usepackage{mathtools}

\usepackage{caption}
\usepackage{amsfonts}
\usepackage{bm}
\usepackage{bbm}
\usepackage{makecell}
\usepackage{multirow}
 \usepackage{amssymb}
\usepackage{txfonts}
\usepackage[T1]{fontenc}
\usepackage{tikz}
\usepackage[scr=dutchcal]{mathalfa}

\usepackage{cite}

\newtheorem{Theorem}{Theorem}
\newtheorem{Proposition}{Proposition}
\newtheorem{Lemma}{Lemma}

\newtheorem{Example}{Example}
\newtheorem{Remark}{Remark}
\newtheorem{Definition}{Definition}

\ifCLASSINFOpdf
\else
\fi
\begin{document}
\title{Typicality Matching for Pairs of Correlated Graphs}

\author{\IEEEauthorblockN{ Farhad Shirani}
\IEEEauthorblockA{Department of Electrical\\and Computer Engineering\\
New York University\\
New York, New York, 11201\\
Email: fsc265@nyu.edu}
\and

\IEEEauthorblockN{Siddharth Garg}
\IEEEauthorblockA{Department of Electrical\\and Computer Engineering\\
New York University\\
New York, New York, 11201\\ 
Email: siddharth.garg@nyu.edu}

\and

\IEEEauthorblockN{Elza Erkip }
\IEEEauthorblockA{Department of Electrical\\and Computer Engineering\\
New York University\\
New York, New York, 11201\\
Email: elza@nyu.edu}
}
%\and
%\IEEEauthorblockN{James Kirk\\ and Montgomery Scott}
%\IEEEauthorblockA{Starfleet Academy\\
%San Francisco, California 96678--2391\\


%Telephone: (800) 555--1212, Fax: (888) 555--1212}


\maketitle

\begin{abstract}
 In this paper, the problem of matching pairs of correlated random graphs with multi-valued edge attributes is considered.  Graph matching problems of this nature arise in several settings of practical interest including social network de-anonymization, study of biological data, web graphs, etc. An achievable region for successful matching is derived by analyzing a new matching algorithm that we refer to as typicality matching. The algorithm operates by investigating the joint typicality of the adjacency matrices of the two correlated graphs. Our main result shows that the achievable region depends on the mutual information between the variables corresponding to the edge probabilities of the two graphs. The result is based on bounds on the typicality of permutations of sequences of random variables that might be of independent interest.
%Finally, conditions when reliable matching is not possible are given by proving a converse result.

\end{abstract}
\section{Introduction}

Graphical models emerge naturally in a wide range of phenomena including social interactions, database systems, and biological systems. In many applications such as DNA sequencing, pattern recognition, and image processing, it is desirable to find algorithms to match correlated graphs. In other applications, such as social networks and database systems, privacy considerations require the network operators to preclude de-anonymization using graph matching by enforcing security safeguards. As a result, there is a large body of work dedicated to characterizing the fundamental limits of graph matching (i.e. to determine the necessary and sufficient conditions for reliable matching), as well as the design of efficient algorithms to achieve these limits.

In the graph matching problem, an agent is given a pair of correlated graphs: i) an `anonymized' unlabeled graph, and ii) a `de-anonymized' labeled graph. The agent's objective is to recover the correct labeling of the vertices in the anonymized graph by matching its vertex set to that of the de-anonymized graph. This is shown in Figure \ref{fig:graph_matching}. This problem has been considered under varying assumptions on the joint graph statistics. Graph isomorphism studied in ~\cite{iso1,iso2,iso4} is an instance of the matching problem where the two graphs are identical copies of one another. Under the Erd\"os- R\'enyi graph model tight necessary and sufficient conditions for graph isomorphism have been derived \cite{ER,wright} and polynomial time algorithms have been proposed ~\cite{iso1,iso2,iso4}. The problem of matching non-identical pairs of correlated Erd\"os-R\'enyi graphs have been studied in ~\cite{corr1,corr2,corr3,corr4, kia_2017, Lyzinski_2016, Asilomar}. Furthermore, graphs with community structure have been considered in ~\cite{corr_com_1,kiavash,efe,Grossglauser}. Seeded versions of the graph matching problem, where the agent has access to side information in the form of partial labelings of the unlabeled graph have also been studied in ~\cite{seed1,seed2,seed3,seed4, Asilomar}. While great progress has been made in characterizing the fundamental limits of graph matching, many of the methods in the literature are designed for specific graph models such as pairs of Erd\H{o}s-R\'enyi graphs with binary valued edges and are not extendable to general scenarios.

 \begin{figure}
 \begin{center}
\includegraphics[width=0.49\textwidth]{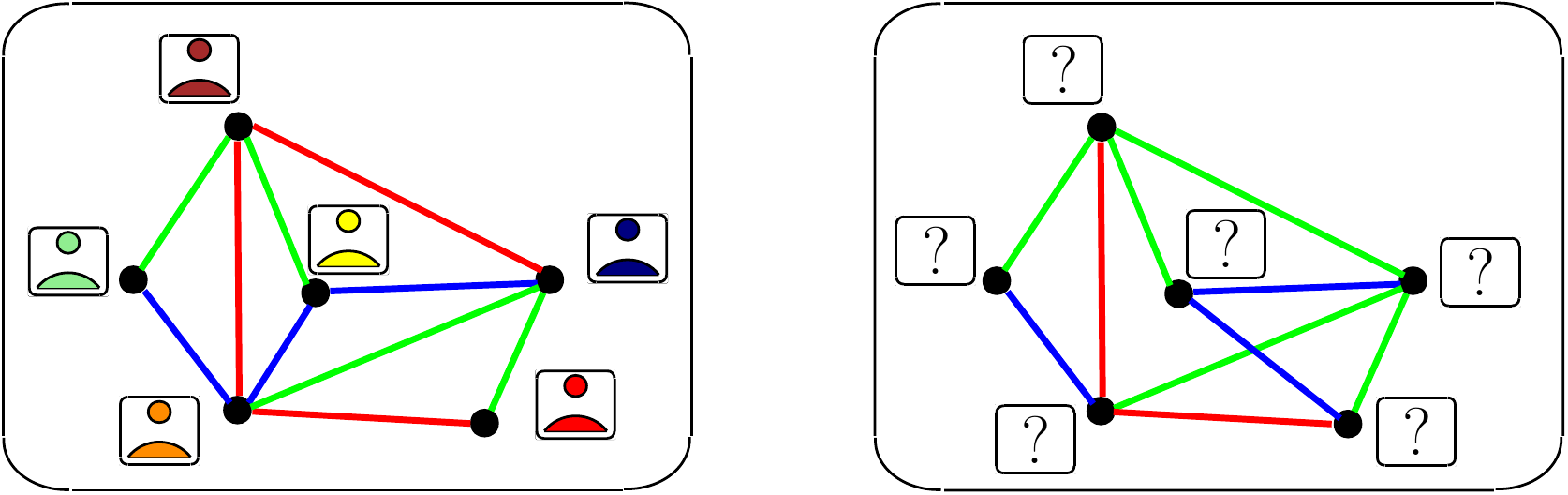}
\caption{An instance of the graph matching problem where the anonymized graph on the right is to be matched to the de-anonymized graph on the left. The edges take values in the set $[0,3]$. The edges with value $0$ are represented by vertex pairs which are not connected and the edges taking values $\{1,2,3\}$ are represented by the colored edges.}
\label{fig:graph_matching}
\end{center}
\end{figure}

In this work, we propose a new approach for analyzing graph matching problems based on the concept of typicality in information theory \cite{csiszarbook}. The proposed approach finds a labeling for the vertices in the anonymized graph which results in a pair of jointly typical adjacency matrices for the two graphs, where typicality is defined with respect to the induced joint statistics of the adjacency matrices. It is shown that if $$I(X_1;X_2)= \omega\left(\frac{\log{n}}{n}\right),$$ 
then it is possible to label the vertices in the anonymized graph such that almost all of the vertices are labeled correctly, where $I(X_1;X_2)$ represents the mutual information between the edge distributions in the two graphs and $n$ is the number of vertices.

The proposed approach is general and leads to a matching strategy which is applicable under a wide range of statistical models. In addition to yielding sufficient conditions for matching correlated random graphs with multi-valued edges, our analysis also includes investigating the typicality of permutations of sequences of random variables which is of independent interest.

The rest of the paper is organized as follows: Section II contains the problem formulation. Section III develops the mathematical machinery necessary to analyze the new matching algorithm. Section IV introduces the typicality matching algorithm and evaluates its performance. Section V concludes the paper. 

\section{Problem Formulation}
In this section, we provide our formulation of the graph matching problem. There are two aspects of our formulation that differ from the ones considered in ~\cite{corr1,corr2,corr3,corr4, kia_2017, Lyzinski_2016, Asilomar}. First, we consider graphs with multi-valued (i.e. not necessarily binary-valued) edges. Second, we consider a relaxed criteria for successful matching. In prior works, a matching algorithm is said to succeed if every vertex in the anonymized graph is matched correctly to the corresponding vertex in the de-anonymized graph with vanishing probability of error. In our formulation, a matching algorithm is successful if a randomly and uniformly chosen vertex in the graph is matched correctly with vanishing probability of error. Loosely speaking, this requires that \textit{almost} all of the vertices be matched correctly.

We consider graphs whose edges take multiple values. Graphs with multi-valued edges appear naturally in various applications where relationships among entities have attributes such as social network de-anonymization, study of biological data, web graphs, etc. An edge which has an attribute assignment is called a \emph{marked edge}. 
The following defines an unlabeled graph whose edges take $l$ different values where $l\geq 2$.
\begin{Definition}
 An $(n,l)$-unlabeled graph $g$ is a structure $(\mathcal{V}_n,\mathcal{E}_n)$, where $n\in \mathbb{N}$ and $l\geq 2$. The set $\mathcal{V}_n=\{v_{n,1},v_{n,2},\cdots,v_{n,n}\}$ is called the vertex set, and the set $\mathcal{E}_n\subset\{(x,v_{n,i},v_{n,j})|x\in [0,l-1], i\in [1,n], j\in[1,n]\}$ is called the marked edge set of the graph. For the marked edge $(x,v_{n,i},v_{n,j})$ the variable `$x$' represents the value assigned to the edge $(v_{n,i},v_{n,j})$ between vertices $v_{n,i}$ and $v_{n,j}$.
% For a given vertex $v_{n,i}, i\in [1,n]$ the set of neighboring vertices is defined as $\mathcal{E}_{n,i}=\{v_{n,j}|(v_{n,i},v_{n,j})\in \mathcal{E}_n\}$. The degree of $v_{n,i}$ is defined as $d_{v_{n,i}}=|\mathcal{E}_{n,i}|$.
\end{Definition}

Without loss of generality, we assume that for any arbitrary pair of vertices $(v_{n,i},v_{n,j})$, there exists a unique $x\in [0,l-1]$ such that $(x,v_{n,i},v_{n,j})\in \mathcal{E}_n$. As an example, for graphs with binary valued edges if the pair $v_{n,i}$ and $v_{n,i}$ are not connected, we write $(0,v_{n,i},v_{n,j})\in \mathcal{E}_n$, otherwise $(1,v_{n,i},v_{n,j})\in \mathcal{E}_n$.

\begin{Definition}
  For an $(n,l)$-unlabeled graph $g=(\mathcal{V}_n,\mathcal{E}_n)$, a labeling is defined as a bijective function $\sigma: \mathcal{V}_n\to [1,n]$.  
The structure $\tilde{g}=(g, \sigma)$ is called an $(n,l)$-labeled graph. For the labeled graph $\tilde{g}$ the adjacency matrix is defined as $\widetilde{G}_{\sigma}=[\widetilde{G}_{\sigma,i,j}]_{i,j\in [1,n]}$ where $\widetilde{G}_{\sigma,i,j}$ is the unique value such that $(\widetilde{G}_{\sigma,i,j},\sigma^{-1}(i),\sigma^{-1}(j))\in \mathcal{E}_n$. The upper triangle (UT) corresponding to $\tilde{g}$ is the structure $\widetilde{G}^{UT}_{\sigma}=[\widetilde{G}_{\sigma,i,j}]_{i<j}$.
\end{Definition}
Any pair of labelings are related through a permutation as described below.
\begin{Definition}
For two labelings $\sigma$ and $\sigma'$, the $(\sigma,\sigma')$-permutation is defined as the bijection $\pi_{(\sigma,\sigma')}$, where:
\begin{align*}
\pi_{(\sigma,\sigma')}(i)=j, \quad \text{if} \quad {\sigma'}^{-1}(j)=\sigma^{-1}(i), \forall i,j\in [1,n].
\end{align*}
\label{def:above}
\end{Definition}
Proposition \ref{prop:1} given bellow follows from Definition \ref{def:above}.
\begin{Proposition}
 Given an $(n,l)$-unlabeled graph $g=(\mathcal{V}_n,\mathcal{E}_n)$ and a pair of arbitrary permutations $\sigma,\sigma'\in S_n$, the adjacency matrices corresponding to $\tilde{g}=(g, \sigma)$ and $\tilde{g}'=(g, \sigma')$ satisfy the following equality:
\begin{align*}
 \widetilde{G}_{\sigma,i,j}=  \widetilde{G}_{\sigma',\pi_{(\sigma,\sigma')}(i),\pi_{(\sigma,\sigma')}(j)}.
\end{align*}
We write $ \widetilde{G}_{\sigma'}=\Pi_{(\sigma,\sigma')}(\widetilde{G}_{\sigma})$, where $\Pi_{(\sigma,\sigma')}$ is an $n^2$-length permutation.
Similarly, we write $ {\widetilde{G'}}^{UT}_{\sigma}=\Pi^{UT}_{(\sigma,\sigma')}(\widetilde{G}^{UT}_{\sigma})$.
\label{prop:1}
\end{Proposition}

\begin{Definition}
Let the random variable $X$ be defined on the probability space $(\mathcal{X},P_X)$, where $\mathcal{X}=[0,l-1]$. A marked Erd\"os-R\'enyi (MER) graph $g_{n,P_X}$ is a randomly generated $(n,l)$-unlabeled graph with vertex set $\mathcal{V}_n$ and edge set $\mathcal{E}_n$, such that
\begin{align*}
 Pr((x,v_{n,i},v_{n,j})\in \mathcal{E}_n)= P_X(x), \forall x\in [1,l-1], v_{n,i},v_{n,j}\in \mathcal{V}_n,
\end{align*}
and edges between different vertices are mutually independent. 

 %In other words, for an arbitrary edge $e=(i,j), i,j\in [1,n]$, we have $P(e\in \mathcal{E})=p$, and the edges generated are mutually independently. 
\end{Definition}

We consider families of correlated pairs of marked labeled Erd\"os-R\'enyi graphs $\underline{g}_{n,P_{n,X_1,X_2}}=(\tilde{g}^1_{n,P_{X_1}},\tilde{g}^2_{n,P_{X_2}})$. 

\begin{Definition}
Let the pair of random variables $(X_1,X_2)$ be defined on the probability space $(\mathcal{X}\times\mathcal{X}, P_{n,X_1,X_2})$, where $\mathcal{X}=[0,l-1]$. A correlated pair of marked labeled Erd\"os-Renyi graphs (CMER)  $\tilde{\underline{g}}_{n,P_{X_1,X_2}}=(\tilde{g}^1_{n,P_{X_1}},\tilde{g}^2_{n,p_{X_2}})$ is characterized by: i) the pair of marked Erd\"os-Renyi graphs $g^i_{n,P_{X_i}}, i\in \{1,2\}$, ii) the pair of labelings $\sigma^i$ for the unlabeled graphs $g^i_{n,P_{X_i}}, i\in \{1,2\}$, and iii) the probability distribution $P_{X_1,X_2}(x_1,x_2), (x_1,x_2) \in \mathcal{X}\times\mathcal{X}$, such that: 
\\1)The pair $g^i_{n,P_{X_i}}, i\in \{1,2\}$ have the same set of vertices $\mathcal{V}_n=\mathcal{V}_n^1=\mathcal{V}_n^2$. 
\\2) For any two marked edges $e^i=(x_i,v^i_{n,s_1},v^i_{n,s_2}), i\in \{1,2\}, x_1,x_2\in [0,l-1]$, we have
 \begin{align*}
 &Pr\left(e^1\in \mathcal{E}_n^1, e^2\in \mathcal{E}_n^2\right)=
 \\&
\begin{cases}
P_{X_1,X_2}(x_1,x_2),& \text{if } \sigma^1(v^1_{n,s_j})=\sigma^2(v^2_{n,s_j}), j\in \{1,2\}\\
P_{X_1}(x_1)P_{X_2}(x_2), & \text{Otherwise}
\end{cases}.
\end{align*} 
% and for an arbitrary edge $e=(v_i,v_j), i,j\in [1,n]$, we have 
%\begin{align*}
% P\left(\mathbbm{1}(e\in \mathcal{E}^1)=\alpha, \mathbbm{1}(e\in \mathcal{E}^2=\beta)\right)=P_{X_1,X_2}(\alpha,\beta), \forall \alpha,\beta\in \{0,1\}^2,
%\end{align*}
\end{Definition}

\begin{Definition}
For a given joint distribution $P_{X_1,X_2}$, a correlated pair of marked partially labeled Erd\"os-Renyi graphs (CMPER)  $\underline{g}_{n,P_{X_1,X_2}}=(\tilde{g}^1_{n,P_{X_1}},{g}^2_{n,P_{X_2}})$  is characterized by: i) the pair of marked Erd\"os-Renyi graphs $g^i_{n,P_{X_i}}, i\in \{1,2\}$, ii) a labeling $\sigma^1$ for the unlabeled graph $g^1_{n,p_{X_1}}$, and iii) a probability distribution $P_{X_1,X_2}(x_1,x_2), (x_1,x_2) \in \mathcal{X}\times\mathcal{X}$,  such that there exists a labeling $\sigma^2$ for the graph $g^2_{n,P_{X_2}}$ for which $(\tilde{g}^1_{n,P_{X_1}},\tilde{g}^2_{n,P_{X_2}})$ is a CMER, where $\tilde{g}^2_{n,P_{X_2}}\triangleq (g^2_{n,P_{X_2}},\sigma^2)$.
\end{Definition}

%\begin{Example}
%Consider the Facebook social network. The friendship relations in the network can be captured using a graph where the vertices represent users in the network and an edge between user $v_i$ and user $v_j$ indicates a friendship relationship between the users. Furthermore, if we consider the pair of graphs corresponding to the Facebook and Twitter social networks, the two resulting graphs are a pair of correlated graphs. Users which are connected in one network are more likely to be connected in the other network as well. In graph matching, we assume that the labeled graph of one network is available while for the other social network, only the unlabeled graph is given. The objective is to find the correct labeling for the unlabeled social network.
%\end{Example}

The following defines a matching algorithm:

\begin{Definition}
 A matching algorithm for the family of CMPERs $\underline{g}_{n,P_{n,X_1,X_2}}=(\tilde{g}^1_{n,P_{n,X_1}},{g}^2_{n,P_{n,X_2}}), n\in \mathbb{N}$ is a sequence of functions 
 $f_n:\underline{g}_{n,P_{n,X_1,X_2}}\mapsto \hat{\sigma}_n^2$ such that $P\left(\sigma_n^2(v^2_{n,J})=\hat{\sigma}_n^2(v^2_{n,J})\right)\to 1$ as $n\to \infty$, where the random variable $J$ is uniformly distributed over $[1,n]$ and   $\sigma^2_n$ is the labeling  for the graph $g^2_{n,P_{n,X_2}}$ for which $(\tilde{g}^1_{n,P_{n,X_1}},\tilde{g}^2_{n,P_{n,X_2}})$ is a CMER, where $\tilde{g}^2_{n,P_{n,X_2}}\triangleq (g^2_{n,P_{n,X_2}},\sigma_n^2)$.
\label{def:algo}
\end{Definition}
Note that in the above definition, for $f_n$ to be a matching algorithm, the fraction of vertices whose labels are matched incorrectly must vanish as $n$ approaches infinity. This is a relaxation of  the criteria considered in \cite{corr1,corr2,corr3,corr4, kia_2017, Lyzinski_2016, Asilomar} where all of the vertices are required to be matched correctly simultaneously with vanishing probability of error as $n\to \infty$.

The following defines an achievable region for the graph matching problem.
\begin{Definition}
 For the graph matching problem, a family of sets of distributions $\widetilde{P}=(\mathcal{P}_n)_{n\in \mathbb{N}}$ is said to be in the achievable region if for every sequence of distributions $P_{n,X_1,X_2}\in \mathcal{P}_n, n\in \mathbb{N}$, there exists a matching algorithm.
 \end{Definition}

\section{Permutations of Typical Sequences}
In this section, we develop the mathematical tools necessary to analyze the performance of the typicality matching strategy. In summary, the typicality matching strategy operates as follows. Given a CMPER $\underline{g}_{n,P_{n,X_1,X_2}}$ the strategy finds a labeling $\hat{\sigma}^2$ such that the pair of adjacency matrices $(\widetilde{G}^{1}_{\sigma^1}, \widetilde{G}^{2}_{\hat{\sigma}^2})$ are jointly typical with respect to $P_{n,X_1,X_2}$, where joint typicality is defined in Section \ref{sec:algo}. Each labeling  $\hat{\sigma}^2$ gives a permutation of the adjacency matrix $\widetilde{G}^{2}_{{\sigma}^2}$. Hence, analyzing the performance of the typicality matching strategy requires bounds on the probability of typicality of permutations of correlated pairs of sequences of random variables. The necessary bounds are derived in this section. The details of typicality matching and its performance are described in Section \ref{sec:algo}.
\begin{Definition}
Let the pair of random variables $(X,Y)$ be defined on the probability space $(\mathcal{X}\times\mathcal{Y},P_{X,Y})$, where $\mathcal{X}$ and $\mathcal{Y}$ are finite alphabets. The $\epsilon$-typical set of sequences of length $n$ with respect to $P_{X,Y}$ is defined as:
\begin{align*}
&A_{\epsilon}^n(X,Y)=
\\&\Big\{(x^n,y^n): \Big|\frac{1}{n}N(\alpha,\beta|x^n,y^n)-P_{X,Y}(\alpha,\beta)\Big|\leq \epsilon, \forall (\alpha,\beta)\in \mathcal{X}\times\mathcal{Y}\Big\},
\end{align*}
where $\epsilon>0$, $n\in \mathbb{N}$, and $N(\alpha,\beta|x^n,y^n)= \sum_{i=1}^n \mathbbm{1}\left((x_i,y_i)=(\alpha,\beta)\right)$.  
\end{Definition}  
We follow the notation used in \cite{isaacs} in our study of permutation groups.
\begin{Definition}
\label{def:perm1}
A permutation on the set of numbers $[1,n]$ is a bijection $\pi:[1,n]\to [1,n]$. The set of all permutations on the set of numbers $[1,n]$ is denoted by $S_n$. 
\end{Definition}
\begin{Definition}
 A permutation $\pi \in S_n, n\in \mathbb{N}$  is called a cycle if there exists $m\in [1,n]$ and $\alpha_1,\alpha_2,\cdots,\alpha_m\in [1,n]$ such that i) $\pi(\alpha_i)=\alpha_{i+1}, i\in [1,m-1]$, ii) $\pi(\alpha_n)=\alpha_1$, and iii) $\pi(\beta)=\beta$ if $\beta\neq \alpha_i, \forall i\in [1,m]$. The variable $m$ is called the length of the cycle. The element $\alpha$ is called a fixed point of the permutation if $\pi(\alpha)=\alpha$. We write $\pi=(\alpha_1,\alpha_2,\cdots,\alpha_m)$.
 The permutation $\pi$ is called a non-trivial cycle if $m\geq 2$. 
\end{Definition}

\begin{Lemma}\cite{isaacs}
 Every permutation $\pi \in S_n, n\in \mathbb{N}$ has a unique representation as a product of disjoint non-trivial cycles.  
\end{Lemma}
\begin{Definition}
\label{def:perm2}
 For a given sequence $y^n\in \mathbb{R}^n$ and permutation $\pi\in S_n$, the sequence $z^n=\pi(y^n)$ is defined as $z^n=(y_{\pi(i)})_{i\in [1,n]}$.\footnote{Note that in Definitions \ref{def:perm1} and \ref{def:perm2} we have used  $\pi$ to denote both a scalar function which operates on the set $[1,n]$ as well as a function which operates on the vector space $\mathbb{R}^n$.}
\end{Definition}
For a correlated pair of independent and identically distributed (i.i.d) sequences $(X^n,Y^n)$ and an arbitrary permutation $\pi\in S_n$, we are interested in bounding the probability $P((X^n,\pi(Y^n))\in A_{\epsilon}^n(X,Y))$. As an intermediate step, we first find a suitable permutation $\pi'$ for which $P((X^n,\pi(Y^n))\in A_{\epsilon}^n(X,Y))\leq P((X^n,\pi'(Y^n))\in A_{\epsilon}^n(X,Y))$.
In our analysis, we make extensive use of the standard permutations defined below.
\begin{Definition}
For a given $n,m,c\in \mathbb{N}$, and $1\leq i_1\leq i_2\leq \cdots\leq i_c \leq n$ such that $n=\sum_{j=1}^ci_j+m$, an $(m,c,i_1,i_2,\cdots,i_c)$-permutation is a permutation in $S_n$ which has $m$ fixed points and $c$ disjoint cycles with lengths $i_1,i_2,\cdots,i_c$, respectively.

The $(m,c,i_1,i_2,\cdots,i_c)$-standard permutation is defined as the $(m,c,i_1,i_2,\cdots,i_c)$-permutation consisting of the cycles $(\sum_{j=1}^{k-1}i_j+1,\sum_{j=1}^{k-1}i_j+2,\cdots,\sum_{j=1}^{k}i_j), k\in [1,c]$. Alternatively, the $(m,c,i_1,i_2,\cdots,i_c)$-standard permutation is defined as:
\begin{align*}
\pi&=(1,2,\cdots,i_1)(i_1+1,i_1+2,\cdots,i_1+i_2)\cdots 
\\&(\sum_{j=1}^{c-1}i_j+1,\sum_{j=1}^{c-1}i_j+2,\cdots,\sum_{j=1}^{c}i_j)(n-m+1)(n-m+2)\cdots (n).
\end{align*}
\label{def:stan_perm}
\end{Definition}
\begin{Example}
The $(2,2,3,2)$-standard permutation is a permutation which has $m=2$ fixed points and $c=2$ cycles. The first cycle has length $i_1=3$ and the second cycle has length $i_2=2$. It is a permutation 
on sequences of length $n=\sum_{j=1}^ci_j+m=3+2+2=7$. The permutation is given by $\pi= (1 2 3)(4 5)(6)(7)$. For an arbitrary sequence $\underline{\alpha}=(\alpha_1,\alpha_2,\cdots,\alpha_n)$, we have:
\begin{align*}
 \pi(\underline{\alpha})=(\alpha_3,\alpha_1,\alpha_2,\alpha_5,\alpha_4,\alpha_6,\alpha_7).
\end{align*}
\end{Example}
\begin{Proposition}
 Let $(X^n,Y^n)$ be a pair of i.i.d sequences defined on finite alphabets. We have:
\\ i) For an arbitrary permutation $\pi\in S_n$, 
 \begin{align*}
 P((\pi(X^n),\pi(Y^n))\in A_{\epsilon}^n(X,Y))=P((X^n,Y^n)\in A_{\epsilon}^n(X,Y)).
\end{align*}
ii)  let $n,m,c,i_1,i_2,\cdots,i_c\in \mathbb{N}$ be numbers as described in Definition \ref{def:stan_perm}. Let $\pi_1$ be an arbitrary $(m,c,i_1,i_2,\cdots,i_c)$-permutation  and let $\pi_2$ be the $(m,c,i_1,i_2,\cdots,i_c)$-standard permutation. Then, 
\begin{align*}
 P((X^n,\pi_1(Y^n))\in A_{\epsilon}^n(X,Y))=P((X^n,\pi_2(Y^n))\in A_{\epsilon}^n(X,Y)).
 \end{align*}
\end{Proposition}
%Please refer to \cite{arxiv_matching_ISIT18}.
\begin{proof}
 The proof of part i) follows from the fact that permuting both $X^n$ and $Y^n$ by the same permutation does not change their joint type. For part ii), it is straightforward to show that there exists a permutation $\pi$ such that $\pi(\pi_1)=\pi_2(\pi)$ \cite{isaacs}. Then the statement follows from part i): 
  \begin{align*}
& P\left(\left(X^n,\pi_1\left(Y^n\right)\right)\in A_{\epsilon}^n\left(X,Y\right)\right)
\\&= P\left(\left(\pi\left(X^n\right),\pi\left(\pi_1\left(Y^n\right)\right)\right)\in A_{\epsilon}^n\left(X,Y\right)\right)
\\&= P\left(\left(\pi\left(X^n\right),\pi_2\left(\pi\left(Y^n\right)\right)\right)\in A_{\epsilon}^n\left(X,Y\right)\right)
\\&\stackrel{(a)}{=} P\left(\left(\widetilde{X}^n,\pi_2\left(\widetilde{Y}^n\right)\right)\in A_{\epsilon}^n\left(X,Y\right)\right)
\\&\stackrel{(b)}{=} P\left(\left(X^n,\pi_2\left(Y^n\right)\right)\in A_{\epsilon}^n\left(X,Y\right)\right),
 \end{align*}
 where in (a) we have defined $(\widetilde{X}^n,\widetilde{Y}^n)=(\pi(X^n),\pi(Y^n))$. and (b) holds since $(\widetilde{X}^n,\widetilde{Y}^n)$ has the same distribution as $(X^n,Y^n)$.
\end{proof}
For a given permutation $\pi\in S_n$, and sequences $(X^n,Y^n)$, define $U_{(\pi)}^n=\pi(Y^n)$. Furthermore, define $Z^{A}_{(\pi),i}=\mathbbm{1}((X_i,U_{(\pi),i})\in A), A \subseteq \mathcal{X}\times \mathcal{Y}$. Define $P_XP_Y(A)=\sum_{(x,y)\in A}P_{X}(x)P_Y(y)$ and $P_{X,Y}(A)=\sum_{(x,y)\in A}P_{X,Y}(x,y)$.

\begin{Theorem}
Let $(X^n,Y^n)$ be a pair of i.i.d sequences defined on finite alphabets $\mathcal{X}$ and $\mathcal{Y}$, respectively. There exists $\zeta>0$ such that for any $(m,c,i_1,i_2,\cdots,i_c)$-permutation $\pi$, and $0<\epsilon<\frac{I(X;Y)}{|\mathcal{X}||\mathcal{Y}|}$:
 \begin{align*}
  P((X^n,\pi(Y^n))\in A_{\epsilon}^n(X,Y))   \leq 2^{-\zeta n(I(X;Y)-|\mathcal{X}||\mathcal{Y}|\epsilon)},
\end{align*}
where $n,m,c,i_1,i_2,\cdots,i_c\in \mathbb{N}$ such that $i_1\geq i_2\geq \cdots \geq i_c,$ and $m<\sqrt{n}$.
\label{th:adv}
\end{Theorem}
The proof is provided in the Appendix. 

 \section{The Typicality Matching Strategy}
 \label{sec:algo} 
In this section, we describe the typicality matching algorithm and characterize its achievable region. Given a CMPER $\underline{g}_{n,P_{n,X_1,X_2}}=(\tilde{g}^1_{n,P_{n,X_1}},{g}^2_{n,P_{n,X_2}})$, the typicality matching algorithm operates as follows. The algorithm finds a labeling $\hat{\sigma}^2$, for which the pair of UT's $\widetilde{G}^{1,UT}_{{\sigma}^1}$ and $\widetilde{G}^{2,UT}_{\hat{\sigma}^2}$ are jointly typical with respect to $P_{n,X_1,X_2}$  when viewed as vectors of length $\frac{n(n-1)}{2}$. Specifically, it returns a randomly picked element $\hat{\sigma}^2$ from the set:
\begin{align*}
 \widehat{\Sigma}=\{\hat{\sigma}^2| (\widetilde{G}^{1,UT}_{{\sigma}^1},\widetilde{G}^{2,UT}_{\hat{\sigma}^2})\in A_{\epsilon}^{\frac{n(n-1)}{2}}\},
\end{align*}
where $\epsilon=\omega(\frac{1}{n})$, and declares $\hat{\sigma}^2$ as the correct labeling. Note that the set $ \widehat{\Sigma}$ may have more than one element. We will show that under certain conditions on the joint graph statistics, all of the elements of $ \widehat{\Sigma}$ satisfy the criteria for successful matching given in Definition \ref{def:algo}. In other words, for all of the elements of $ \widehat{\Sigma}$  the probability of incorrect labeling for any given vertex is arbitrarily small for large $n$.
\begin{Theorem}
For the typicality matching algorithm, a given family of sets of distributions $\widetilde{P}=(\mathcal{P}_n)_{n\in \mathbb{N}}$ is achievable, if for every sequence of distributions $P_{n,X_1,X_2}\in \mathcal{P}_n, n\in \mathbb{N}$:
\begin{align}
I(X_1;X_2)=\omega(\frac{\log{n}}{n}),
\label{eq:th21}
\end{align}
\label{th:3}
provided that $P_{n,X_1,X_2}$ is bounded away from 0 as $n\to \infty$.
% Particularly, if $X_1$ and $X_2$ are binary random variables, then $\widetilde{P}=(\mathcal{P}_n)_{n\in \mathbb{N}}$ is achievable, if for every sequence of distributions $P_{n,X_1,X_2}\in \mathcal{P}_n, n\in \mathbb{N}$
% \begingroup\makeatletter\def\f@size{9}\check@mathfonts
%\begin{align}
%\Omega(\frac{\log{n}}{n})=(P_{n,X_1,X_2}(0,0)P_{n,X_1,X_2}(1,1)-P_{n,X_1,X_2}(0,1)P_{n,X_1,X_2}(1,0))^2.
%\label{eq:kia}
%\end{align}
%\endgroup
\end{Theorem}
The proof is provided in the Appendix.
\begin{Remark}
For graphs with binary valued edges, Theorem \ref{th:3} provides bounds on the condition for successful matching which improve upon the bound given in (\cite{kia_2017} Theorem 1). It should be noted that a stronger definition for successful matching is used in \cite{kia_2017}.
\end{Remark}
%\begin{Theorem}
% For the typicality matching algorithm, the sequence of distributions $P_{n,X_1,X_2}\in \mathcal{P}_n, n\in \mathbb{N}$ is not achievable if
%\begin{align}
%O(\frac{\log{n}}{n})=I_{P_{n,X_1,X_2}}(X_1;X_2).
%\label{eq:th21}
%\end{align}
%\end{Theorem} 
 \section{Conclusion}
 We have introduced the typicality matching algorithm for matching pairs of correlated graphs. The probability of typicality of permutations of sequences of random variables has been investigated.  An achievable region for the typicality matching algorithm has been derived. The region characterizes the conditions for successful matching both for graphs with binary valued edges as well graphs with finite-valued edges.
\appendix
\subsection{Proof of Theorem \ref{th:adv}}
The proof builds upon some of the results in \cite{chen}. Fix an integer $t\geq 2$. We provide an outline of the proof when $m=0$ and $i_1\leq t$.
Let $A=\{(x,y)\in \mathcal{X}\times\mathcal{X}\big| P_XP_Y(x,y)<P_{X,Y}(x,y)\}$ and $\epsilon\in [0, \min_{(x,y)\in \mathcal{X}\times\mathcal{X}}(|P_{X,Y}(x,y)-P_XP_{Y}(x,y)|)]$. Note that
\begin{align*}
&Pr((X^n,\pi(Y^n))\in A_{\epsilon}^n(X,Y))\leq
 \\&Pr\Big(\Big(\bigcap_{(x,y)\in A}\big\{\frac{1}{n}\sum_{i=1}^nZ^{\{(x,y)\}}_{(\pi),i}>P_{X,Y}(x,y)-{\epsilon}\big\}\Big)
 \bigcap
 \\&\Big(\bigcap_{(x,y)\in A^c}\big\{\frac{1}{n}\sum_{i=1}^nZ^{\{(x,y)\}}_{(\pi),i}<P_{X,Y}(x,y)-{\epsilon}\big\}\Big) 
  \Big)
\end{align*}
 For brevity let $\alpha_{x,y}=\frac{1}{n}\sum_{i=1}^nZ^{\{(x,y)\}}_{(\pi),i}, x,y\in \mathcal{X}$ and let $$T^{(x,y)}_j=\frac{1}{n}\sum_{k=1}^{i_j}Z^{\{(x,y)\}}_{(\pi),k}, j\in [1,c].$$ 
Also, define $\bar{c}=\frac{n}{t}$ and $t_{x,y}=\log_e{\frac{P_{X,Y}(x,y)}{P_X(x)P_Y(y)}}$.
 Then, 
\begin{align*}
   &P\Big(\Big(\bigcap_{(x,y)\in A}\big\{\bar{c}\alpha_{x,y}>\bar{c}P_{X,Y}(x,y)-{\epsilon}\big\}\Big)
 \bigcap
 \\&\Big(\bigcap_{(x,y)\in A^c}\big\{\bar{c}\alpha_{x,y}<\bar{c}P_{X,Y}(x,y)-{\epsilon}\big\}\Big) 
  \Big)\\
&=  P\Big(\bigcap_{(x,y)\in \mathcal{X}\times\mathcal{Y}}\big\{e^{\bar{c}t_{x,y}\alpha_{x,y}}>e^{\bar{c}t_{x,y}P_{X,Y}(x,y)-{\epsilon}}\big\}),
\end{align*}
where we have used the fact that by construction:
\begin{align}
\begin{cases}
t_{x,y}>0 \qquad& \text{if} \qquad (x,y)\in A\\
t_{x,y}<0 &\text{if} \qquad (x,y)\in A^c.
\end{cases} 
\label{eq:cases}
\end{align}
So,
\begin{align}
 \nonumber  &P\Big(\Big(\bigcap_{(x,y)\in A}\big\{\bar{c}\alpha_{x,y}>\bar{c}P_{X,Y}(x,y)-{\epsilon}\big\}\Big)
 \bigcap
 \\\nonumber&\Big(\bigcap_{(x,y)\in A^c}\big\{\bar{c}\alpha_{x,y}<\bar{c}P_{X,Y}(x,y)-{\epsilon}\big\}\Big) 
  \Big)\\
&\stackrel{(a)}{\leq}  P\Big(\prod_{(x,y)\in \mathcal{X}\times\mathcal{Y}}e^{\bar{c}t_{x,y}\alpha_{x,y}}>\prod_{(x,y)\in \mathcal{X}\times\mathcal{Y}}e^{\bar{c}t_{x,y}P_{X,Y}(x,y)-{\epsilon}}\Big)\\
&\stackrel{(b)}{\leq} e^{-\sum_{x,y}\bar{c}t_{x,y}P_{X,Y}(x,y)-{\epsilon}}\mathbb{E}(\prod_{x,y}e^{\bar{c}t_{x,y}\alpha_{x,y}})
\\\nonumber&= 
e^{-\sum_{x,y}\bar{c}t_{x,y}P_{X,Y}(x,y)-{\epsilon}}\mathbb{E}(e^{\frac{\bar{c}}{n}\sum_{j=1}^{{c}}\sum_{x,y}t_{x,y}T^{(x,y)}_j})
\\&\stackrel{(c)}{=} e^{-\sum_{x,y}\bar{c}t_{x,y}P_{X,Y}(x,y)-{\epsilon}}\prod_{j=1}^{{c}}\mathbb{E}(e^{\frac{1}{t}\sum_{x,y}t_{x,y}T_j^{\{(x,y)\}}}),
\label{eq:temp3}
\end{align}
where in (a) we have used the fact that the exponential function is increasing and positive, (b) follows from the Markov inequality and (c) follows from the independence of $T_j^{\{(x,y)\}}$ and $T_i^{\{(x,y)\}}$ when $i\neq j$ for arbitrary $(x,y)$ and $(x',y')$.
Next, we investigate the term $\mathbb{E}(e^{\frac{1}{t}\sum_{x,y}t_{x,y}T_j^{\{(x,y)\}}})$. Note that by Definition, $\sum_{x,y\in \mathcal{X}}T^{\{(x,y)\}}_{j}=i_j, \forall j\in [1,c]$.
Define $S^{\{(x,y)\}}_j=\frac{1}{t}T^{\{(x,y)\}}_{j}, j\in [1,c]$.
 Let $\mathcal{B}=\{(s^{\{(x,y)\}}_{j})_{j\in [1,c], x,y\in \mathcal{X}}: \sum_{x,y\in \mathcal{X}}s^{\{(x,y)\}}_{j}=\frac{i_j}{t}, \forall j\in [1,c]\}$ be the set of possible values taken by $(S^{\{(x,y)\}}_{j})_{j\in [1,c], x,y\in \mathcal{X}}$. Note that:
\begin{align}
 &\mathbb{E}(e^{\sum_{x,y}t_{x,y}S^{\{(x,y)\}}_{j}}) \nonumber
\\&=\sum_{(s^{\{(x,y)\}}_{j})_{i\in [1,n], x,y\in \mathcal{X}}\in \beta}
P((s^{\{(x,y)\}}_{j})_{j\in [1,c], x,y\in \mathcal{X}})e^{\sum_{x,y}t_{x,y}s^{\{(x,y)\}}_{j}}.
\label{eq:temp}
\end{align}
For a fixed vector $(s^{\{(x,y)\}}_{j})_{j\in [1,c], x,y\in \mathcal{X}}\in \beta$, let $V_{(x,y)}$ be defined as the random variable for which $P(V=t_{(x,y)})=s^{\{(x,y)\}}_{j}, x,y\in \mathcal{X} $ and $P(V=0)=1-\frac{i_j}{t}$ (note that $P_V$ is a valid probability distribution). From Equation \eqref{eq:temp}, we have:
\begin{align}
\nonumber  &\mathbb{E}(e^{\sum_{x,y}t_{x,y}S^{\{(x,y)\}}_{j}})
=\sum_{(s^{\{(x,y)\}}_{j})_{j\in [1,c], x,y\in \mathcal{X}}\in \beta}
P((s^{\{(x,y)\}}_{j})_{j\in [1,c], x,y\in \mathcal{X}})e^{\mathbb{E}(V_{x,y})}
\\\nonumber&\leq \sum_{(s^{\{(x,y)\}}_{j})_{j\in [1,c], x,y\in \mathcal{X}}\in \beta}
 P((s^{\{(x,y)\}}_{j})_{j\in [1,c], x,y\in \mathcal{X}})\mathbb{E}(e^{V_{x,y}})
 \\\nonumber&=\sum_{(s^{\{(x,y)\}}_{j})_{j\in [1,c], x,y\in \mathcal{X}}\in \beta}
P((s^{\{(x,y)\}}_{j})_{j\in [1,c], x,y\in \mathcal{X}})(1-\frac{i_j}{t}+\sum_{x,y}s^{\{(x,y)\}}_{j}e^{t_{x,y}})
\\\nonumber&=(1-\frac{i_j}{t}+ \sum_{x,y} \mathbb{E}(S^{\{(x,y)\}}_{(\pi),i})e^{t_{x,y}})
\\&=(1-\frac{i_j}{n}+ \sum_{x,y}\frac{i_j}{t}P_X(x)P_Y(y)e^{t_{x,y}}).
\label{eq:temp2}
\end{align}
We replace $t_{x,y}, x,y\in \mathcal{X}\times\mathcal{Y}$ with $\log_e{\frac{P_{X,Y}(x,y)}{P_X(x)P_Y(y)}}$. From Equation \eqref{eq:temp2}, we conclude that $\mathbb{E}(e^{\sum_{x,y}t_{x,y}S^{\{(x,y)\}}_{j}})=1$. From Equation \eqref{eq:temp3}, we have:
\begin{align*}
 &P\Big(\Big(\bigcap_{(x,y)\in A}\big\{\alpha_{x,y}>nP_{X,Y}(x,y)-{\epsilon}\big\}\Big)
 \bigcap
 \\&\Big(\bigcap_{(x,y)\in A^c}\big\{\alpha_{x,y}<nP_{X,Y}(x,y)-{\epsilon}\big\}\Big) 
  \Big)\\
  &=e^{-\frac{n}{t}\sum_{x,y}(P_{X,Y}(x,y)\log_e{\frac{P_{X,Y}(x,y)}{P_X(x)P_Y(y)}}-{\epsilon})}
  \\&=e^{-\frac{n}{t}(I(X;Y)-|\mathcal{X}||\mathcal{Y}|{\epsilon}}).
\end{align*}
In the next step, we prove the theorem when $i_1>t$ and $m=0$. The proof is similar to the previous case. Following the steps above, we get: 
 \begin{align*}
  &Pr((X^n,\pi(Y^n))\in A_{\epsilon}^n(X,Y))  =  Pr(V(P_{X,Y},\hat{P})\leq \epsilon) \nonumber
\\&\stackrel{(c)}{=} e^{-\sum_{x,y}\bar{c}t_{x,y}P_{X,Y}(x,y)-{\epsilon}}\prod_{j=1}^{{c}}\mathbb{E}(e^{\frac{1}{t}\sum_{x,y}t_{x,y}T_j^{\{(x,y)\}}})
\end{align*}
Assume that $i_1>t$, then we `break' the cycle into smaller cycles as follows:

\begin{align*}
 &\mathbb{E}(e^{\frac{1}{t}\sum_{x,y}t_{x,y}T_1^{\{(x,y)\}}})=
 \\&\mathbb{E}_{X_t}(\mathbb{E}(e^{\frac{1}{t}\sum_{x,y}t_{x,y}T_1^{\{(x,y)\}}})|X_t)
 \\&=\mathbb{E}_{X_t}(\mathbb{E}(e^{\frac{1}{t}\sum_{x,y}t_{x,y}T_1^{' \{(x,y)\}}}|X_t)\mathbb{E}(e^{\frac{1}{t}\sum_{x,y}t_{x,y}T_1^{'' \{(x,y)\}}})|X_t),
\end{align*}
where, $T_1^{' \{(x,y)\}}=\frac{1}{n}\sum_{k=1}^{t}Z^A_{(\pi),k}$ and $T_1^{'' \{(x,y)\}}=\frac{1}{n}\sum_{k=t+1}^{i_1}Z^A_{(\pi),k}$. We investigate $\mathbb{E}(e^{\frac{1}{t}\sum_{x,y}t_{x,y}T_1^{' \{(x,y)\}}}|X_t)$. Define, $S^{' \{(x,y)\}}_1=\frac{1}{t}T^{' \{(x,y)\}}_{1}$. Then, 
\begin{align*}
 &\mathbb{E}(e^{\sum_{x,y}t_{x,y}S^{'\{(x,y)\}}_{j}}|X_t)
\\&=(1-\frac{i_j}{n}+ \sum_{x,y}(\frac{i_j-1}{t}P_X(x)P_Y(y)+\frac{1}{t}P_X(x)P_{Y|X}(y|X_t))e^{t_{x,y}} )
\\&=1.
\end{align*}
The theorem is proved by the repetitive application of the above arguments. 
\subsection{Proof of Theorem \ref{th:3}}
First, note that
\begin{align*}
P((\widetilde{G}^{1,UT}_{{\sigma}^1},\widetilde{G}^{2,UT}_{{\sigma}^2})\in A_{\epsilon}^{\frac{n(n-1)}{2}})\to 1 \quad \text{as}\quad n\to \infty.
\end{align*}
So, $P(\widehat{\Sigma}=\phi)\to 0$ as $n\to \infty$ since the correct labeling is a member of the set $\widehat{\Sigma}$.
Let $(\lambda_n)_{n\in \mathbb{N}}$ be an arbitrary sequence of numbers such that $\lambda_n= \Theta(n)$. We will show that the probability that a labeling in $\widehat{\Sigma}$ labels $\lambda_n$ vertices incorrectly goes to $0$ as $n\to \infty$. 
Define the following:
\begin{align*}
 \mathcal{E}=\{{\sigma'}^2\Big| ||\sigma^2-{\sigma'}^2||_1\geq \lambda_n\},
\end{align*}
where $||\cdot||_1$ is the $L_1$-norm. The set $\mathcal{E}$ is the set of all labelings which match more than $\lambda_n$ vertices incorrectly.

We show the following:
\begin{align*}
 P(\mathcal{E}\cap \widehat{\Sigma}\neq \phi)\to 0, \qquad \text{as} \qquad n\to \infty.
 \end{align*}
Note that:
\begin{align*}
  &P(\mathcal{E}\cap \widehat{\Sigma}\neq \phi)
  = P(\bigcup_{{\sigma'}^2: ||\sigma^2-{\sigma'}^2||_1\geq \lambda_n}\{{\sigma'}^2\in  \widehat{\Sigma}\})
  \\&\stackrel{(a)}{\leq} \sum_{i=\lambda_n}^{n}\sum_{{\sigma'}^2: ||\sigma^2-{\sigma'}^2||_1=i} P({\sigma'}^2\in  \widehat{\Sigma})
  \\&\stackrel{(b)}{=} \sum_{i=\lambda_n}^{n}\sum_{{\sigma'}^2: ||\sigma^2-{\sigma'}^2||_1=i}
   P((\widetilde{G}^{1,UT}_{{\sigma}^1},\Pi_{\sigma^2,{\sigma'}^2}(\widetilde{G}^{2,UT}_{{\sigma}^2}))\in A_{\epsilon}^{\frac{n(n-1)}{2}})
\\&\stackrel{(c)}{\leq} \sum_{i=\lambda_n}^{n}\sum_{{\sigma'}^2: ||\sigma^2-{\sigma'}^2||_1=i}
   2^{-({\frac{n(n-1)}{2}-\frac{\lambda_n(\lambda_n-1)}{2}})(I(X_1;X_2)-|\mathcal{X}||\mathcal{Y}|\epsilon)}
   \\&\stackrel{(d)}{=}  \sum_{i=\lambda_n}^{n} {n \choose i}(!i)
  2^{-\Theta(n^2)(I(X_1;X_2)-\frac{\epsilon}{2})}
  \\&\leq n^ne^{-\Theta(n^2)(I(X_1;X_2)-|\mathcal{X}||\mathcal{Y}|\epsilon)}
  \\&\leq 2^{-\Theta(n^2)(I(X_1;X_2)-|\mathcal{X}||\mathcal{Y}|\epsilon)-\Theta(\frac{\log{n}}{n}))},
\end{align*}
where (a) follows from the union bound, (b) follows from the definition of $ \widehat{\Sigma}$ and Proposition \ref{prop:1}, in (c) we have used Theorem \ref{th:adv} and the fact that $||\sigma^2-{\sigma'}^2||_1$ means that $\Pi_{\sigma^2,{\sigma'}^2}$ has $\frac{\lambda_n(\lambda_n-1)}{2}$ fixed points, in (d) we have denoted the number of derangement of sequences of length $i$ by $!i$. Note that the right hand side in the last inequality approaches 0 as $n\to \infty$ as long as \eqref{eq:th21} holds since $\epsilon=O(\frac{\log{n}}{n})$.

 %Finally, a converse result was proved which provides conditions under which successful matching is not possible.


% \fi
\IEEEpeerreviewmaketitle

%\textbf{Future Directions:}
%
%1) One could potentially use partial set covers instead of full covers to improve the results (e.g. to prove the sufficiency of a logarithmic number of queries). Unfortunately a complete analysis of the size of such covers seems unavailable. Some references include ~\cite{PSCbook, appPSC}.
%\\2) There seems to be a straightforward generalization of the result to graphs with community structure. In this setting, the community with the minimum edge probability acts as the bottleneck for the expected number of inquiries. 
%\\3) One might consider the case when $s_1,s_2\neq 1$. It seems that the solution would have similar performance to the above where $s_1=s_2=1$ as the number of users becomes asymptotically large.
%\\4) The application of minimum set covers to graph matching problems seems feasible. I did not find any prior works which looked at these applications yet.
%

\bibliographystyle{unsrt}

\end{document}